\RequirePackage{fix-cm}
\documentclass[smallextended]{svjour3}       
\smartqed  
\usepackage{graphicx}
\usepackage{tikz}
\usepackage[ruled]{algorithm2e}

\SetAlFnt{\small}
\SetAlCapFnt{\small}
\SetAlCapNameFnt{\small}
\SetAlCapHSkip{0pt}
\IncMargin{-\parindent}
\usepackage{amsfonts,amssymb,bbm,amstext,amsmath}
\def\np{\bf NP}

\def\ds{\displaystyle}

\newcommand{\soph}[1]{\mbox{soph}_{#1}}
\newcommand{\nstoch}[1]{\mbox{nstoch}_{#1}}

\newcommand{\ld}[1]{\mbox{depth}_{#1}}
\newcommand{\halttime}{\mbox{time}}

\newcommand{\ovl}[1]{\overline{#1}}

\newcommand{\Dom}[1]{\mbox{Dom}{#1}}

\newcounter{oldTheorem}

\begin{document}

\title{Sophistication vs Logical Depth
}


\author{Lu\'is Antunes         \and
        Bruno Bauwens \and
        Andr\'e Souto \and
        Andreia Teixeira
}

\authorrunning{L. Antunes et al.} 

\institute{Lu\'is Antunes \at
              Faculdade de Ciências da Universidade do Porto\\
              SQIG-Instituto de Telecomunica\c{c}\~oes \\
              R.Campo Alegre,1021/1055 \\ 
              4169 - 007 Porto - Portugal \\
              Tel.: +351 220 402 929\\
              Fax: +351 220 402 950 \\
              \email{lfa@dcc.fc.up.pt}           
           \and
           Bruno Bauwens \at
	   Faculty of Computer Science - Research University Higher School of Economics,
	   Kochnovskiy Proezd 3, 125319 Moscow
           \and
           Andr\'e Souto \at
              Departamento de Matemática - Instituto Superior Técnico - Universidade de Lisboa and SQIG-Instituto de Telecomunica\c{c}\~oes
           \and
           Andreia Teixeira \at
              DCC-FCUP and SQIG-Instituto de Telecomunica\c{c}\~oes
}

\date{Received: date / Accepted: date}

\maketitle

\begin{abstract}
Sophistication and logical depth are two measures that express how complicated the structure in a string is. Sophistication is defined as the minimal complexity of a computable function that defines a two-part description for the string that is shortest within some precision; the second can be defined as the minimal computation time of a program that is shortest within some precision.

We show that the Busy Beaver function of the sophistication of a string exceeds its logical depth
with logarithmically bigger precision, and that logical depth exceeds the Busy Beaver function of
sophistication with logarithmically bigger precision. 
  We also show that sophistication is unstable in its precision: constant variations can change its value by a linear term in the length of the string.

\end{abstract}

Solomonoff~\cite{Sol64}, Kolmogorov~\cite{kol65} and Chaitin~\cite{Cha66} independently defined a
measure of information contained in a bit string~$x$ as the length of a shortest program that
produces~$x$ on a universal Turing machine.  This measure, usually represented by~$C(x)$, is called
Kolmogorov complexity.  
Kolmogorov complexity does not express whether the string contains sophisticated structure.  For
example, consider for some $n$ a randomly generated $n$-bit string.  With high probability the
complexity is about $n$ and the string has no (complicated) structure.  On the other hand, the
$(2^n-1)$-bit string representing the Halting problem for programs of length less than $n$ has also
complexity close to~$n$ but has very complicated structure.  Informally, ``sophistication of
structure'' can be measured by the minimal computation time of a program modeling the structure or
by the minimal size of a program that models the structure.

The first notion is Bennett's \em logical depth\em~\cite{ben88}. At significance level~$c$, it is defined as the minimal 
time to compute $x$ by a program $p$ that is $c$-incompressible on a universal prefix-free Turing
 machine $U$ (of some type), {\em i.e.} $C_U(p) \ge |p| - c$.  Bennett~\cite{ben88} showed that this measure is
closely related to the minimal time for which some time-bounded version of algorithmic probability
converges within a factor $2^{-c}$. We will use the following simpler variant (which is closely
related with the previous one, see Section \ref{sec:prel}):

\begin{quotation} 
{\it ``The time required to compute $x$ by a program no more than $c$~bits longer than a shortest program.''}
\end{quotation}

Examples of strings that are non-deep according to this definition are the random strings and the efficiently computable ones.
In~\cite{afmv06}, this notion was used to show that if the complexity class $\np$ reduces to a sequence for which every initial segment is not deep, up to ``polylog'' precision in the length of the string, then the polynomial time hierarchy collapses. In particular, it would imply a collapse if $\np$ reduces to a sparse or to a random set.

Koppel~\cite{koppel87} defined a different notion of depth for infinite sequences based on some variant of monotone Kolmogorov
complexity.  The class of deep sequences is defined by the ones for which the depth of initial
segments is not bounded by a computable function of their length.  In particular, the set of such
sequences is disjoint from the set of random ones, and hence, they define a set of measure zero.
Lutz~\cite{DepthAndReducibility} showed that deep sequences contain useful information in the
following computational sense: the class of sequences that truth-table reduces to them has non-zero
measure in the class of computable sequences.

Kolmogorov~\cite{KolmogorovTallinn,KolmogorovAbstract} defined for each string the notion of structure function dividing a shortest program for a string in two parts -- one part accounting for useful regularities and another accounting for the remaining information presented in the string -- in such a way that this two-part description is as small as the shortest one-part description. 
He represented the regularities in the string by finite sets. Later, Koppel~\cite{koppel87,koppel88,koppel91} expressed regularities as monotone computable functions and called the minimal complexity of the function defining a shortest two-part code {\it sophistication}. Following Koppel's work, Li and Vit\'anyi~\cite[p. 100]{livi} and independently Antunes and Fortnow~\cite{soph} revisited the notion of sophistication considering computable functions (that are not necessarily monotone).
It was observed that there are strings with near maximum sophistication, and such strings encode the halting problem for smaller programs.
Furthermore, in~\cite{soph}  coarse sophistication was introduced, and it was shown
that it is roughly equivalent to a variation of logical depth 
based on the Busy Beaver function. In Section~\ref{sec:overview}, we present a more detailed overview of the literature on these measures of sophistication. %

Sophistication and logical depth are conceptually very different since the former measures program lengths while the latter running times. In order to establish a relationship between these measures, we rescale logical depth from running time to program length using the Busy Beaver function. In this scenario, we prove that, up to logarithmic changes of the significance of both measures, they are equal up to logarithmic terms. 
From this, we conclude that all sophistication measures defined using Kolmogorov complexity, are equivalent in this sense.
A closely related result was previously shown in~\cite[Theorems 3.1.21 and 3.3.3]{bruno}. 
Although, using a very technical but closely related scaling function based on the convergence time of Chaitin Omega numbers.
We also study the stability of sophistication under changes of significance. 
From \cite[Theorem IV.4]{vv04}, one concludes that a logarithmic change of the significance can change sophistication maximally ({\em i.e.} almost $|x|$).
We show this also holds for constant changes of the significance.

\section{Definitions and results}\label{sec:prel}

For a string $x$, let $|x|$ be the length of $x$. 
For each Turing machine, we associate a partial function $U$ that maps pairs of strings to strings. 
We fix a reference Turing machine $U$ that is universal in the following sense: for any other
machine $V$, there is a string $w_V$ such that $U(w_Vp,y) = V(p,y)$ if $V(p,y)$ is defined. 
If $y$ is the empty string we write $U(p)$ rather than $U(p,y)$.

The {\em Kolmogorov complexity} of $x$ is defined as 
\[ 
  C(x)= \min_{p}\{|p| : U(p)= x \}.
\]
Note that changing the universal machine $U$ affects Kolmogorov complexity  by less than an
additive constant.

\bigskip

Koppel~\cite{koppel87}, using monotone functions as a model, defined sophistication for infinite
strings. Later, Li and Vit\'anyi~\cite{LiVi97} and Antunes and Fortnow~\cite{soph} independently
simplified Koppel's definition of sophistication for finite strings, using computable functions
(that are not necessarily monotone). 

\begin{definition}[as in \cite{soph}]\label{AF}
  Let $c$ be integer. 
The \em sophistication \em of a string $x$ with significance $c$  is:
\[
\soph{c}(x)=\min_p\left\{|p|: \begin{array}{l}\mbox{$U(p,d)$ is defined for all $d$}\\ \mbox{ and there is a~$d$ s.t. } U(p,d)=x
  \\ \mbox{ and }|p|+|d|\leq C(x)+c\end{array}\right\} .
  \] 
  If no such $p$ exists, then $\soph{c}(x) = +\infty$.
\end{definition}

Clearly, sophistication is non-increasing in $c$. For $c = |x| + O(1)$, sophistication is bounded by
$O(\log |x|)$. It might happen that sophistication is finite for negative $c$, however one can show
that finite sophistication implies $-c \le O(\log |x|)$. 

\bigskip

Bennett~\cite{ben88} defined the $c$-significant logical depth of an object $x$ as the time required
by a prefix-free machine to generate $x$ with a program $p$ that is $c$-incompressible ({\em i.e.} 
$K(p) \ge |p|-c$, where $K$ stands for the complexity on a universal prefix-free machine).  
Our results are related to a more intuitive version of logical depth (also
discussed in~\cite{ben88}). 
In Appendix~\ref{sec:bbdepthInvariant}, we explain why this notion is sufficiently close
to Bennett's definition.
Let $\halttime(p)$ be the number of computation steps made by $U$ on input $x$ to reach a halting state.

\begin{definition}[Logical depth]\label{logicaldepth}
For any $c \ge 0$, the \em logical depth \em of a string $x$ at significance level $c$ is
\[
\ld{c}(x) = \min\left\{\halttime(p) : |p| \leq C(x) + c \textnormal{ and } U(p) = x\right\} .
\] 
\end{definition}
Note that depth is always finite (for $c\ge 0$). For $c<0$ let $\ld{c}(x) = +\infty$.
One can, scale down the running time to program length using the inverse Busy Beaver function 
\[bb(n) = \min\left\{|p|:U(p) \mbox{ halts and } \halttime(p) \ge n\right\}.\] 
The Busy Beaver logical depth is simply the inverse Busy Beaver function of the 
logical depth. By definition, this equals the minimal complexity of an
upper bound of the logical depth.

\begin{definition}
  The \em Busy Beaver logical depth \em of $x$ with significance $c$ is:
  \begin{eqnarray*}
    \ld{c}^{bb}(x) &=& bb(\ld{c}(x)) \\
    &=& \min_{p,q}\left\{|q|:\begin{array}{l} |p| \leq C(x) + c \mbox{ and } U(p)=x\\ \mbox{and } \halttime(p) \le \halttime(q)\end{array}\right\}
  \end{eqnarray*}
\end{definition}
From the definition it is easy to see that $\ld{c}^{bb}(x) \le C(x) \le |x| + O(1)$. 
Clearly, $\ld{c}(x)$ is non-increasing in $c$. 
For some machines $U$ we have $\ld{c}(x) \ge |x|$ for all $x$, for example, if every halting program 
on $U$ always scans the full input. The following lemma shows that changing two such machines changes 
the Busy Beaver logical depth to a function that is close.

\begin{lemma}\label{lem:depthIsInvariant}
  For all universal\footnote{In fact the proof only requires that $U$ and $V$ are
    optimal, {\em i.e.} for all machines $W$ there exist $c_W$ 
    such that $C_U(x) \le C_W(x) + c_W$ and similarly for $V$.} 
  Turing machines $U$ and $V$, there exist a constant~$c'$ 
  such that for all $c$ and $x$: $\ld{c,U}(x) \ge |x|$ [no Busy Beaver here!] implies
  \[\ld{c+c',V}^{bb}(x) \le \ld{c,U}^{bb}(x) + c'.\]
\end{lemma}
We postpone the proof of this lemma to the Appendix~\ref{sec:bbdepthInvariant}. 
Let the {\em upper graph} of a function $f$ be $\{(n,m): m \ge f(n)\}$.
Let the distance between two points $(n,m)$ and $(n',m')$ be $\max(|n-n'|, |m-m'|)$.
\begin{definition}\label{def:close}
Two functions $f$ and $g$ are  {\em $c$-close} if the upper graphs of these functions are in a $c$-neighbourhood of each other.
\end{definition}
If $f$ and $g$ are non-increasing, this is equivalent to $f(n+c)
\le g(n) + c$ and $g(n+c) \le f(n) + c$. The previous lemma shows that the depth function of 
all universal machines $U$ with $\ld{c,U}(x) \ge |x|$ are $O(1)$-close.

\bigskip

The first main result of the paper states that sophistication and logical depth are $O(\log |x|)$-close.

\begin{theorem} \label{th:ldepthVsSoph}
  For a fixed $x$, the functions $\ld{c}^{bb}(x)$ and $\soph{c}(x)$ are $O(\log |x|)$-close, {\em i.e.},  for some $e$ and
  for all $c$ and $x$ with $|x| \ge e$:
  \begin{eqnarray*}
     \ld{c + e\log |x|}^{bb}(x) &\le& \soph{c}(x)+e\log |x| \\
     \soph{c + e\log |x|}(x) &\le& \ld{c}^{bb}(x) + e\log |x|.
  \end{eqnarray*}
\end{theorem}
\newcounter{ldepthVsSoph}
\setcounter{ldepthVsSoph}{\value{theorem}}

In Theorem~\ref{th:ldepthDiffersSoph} in section~\ref{sec:ldepthDiffersSoph} it is shown 
that the margin in the significance cannot be made constant, and hence, depth and sophistication are not $O(1)$-close.
The second main result states that for a fixed string $x$, the sophistication function is unstable in its
significance; more precisely, for some $x$ and~$c$, small changes of $c$
can result in large changes of $\soph{c}(x)$.

\begin{theorem}\label{th:sophUnstable}
For some $e$ and for large $c$ there are infinitely many $x$ such that\footnote{
  For any $\varepsilon>0$, we can replace the term $\tfrac{3}{4}|x|$ by $(1-\varepsilon)|x|$ 
  if the significance of the second sophistication term is replaced by $c + O(\log (c/\varepsilon))$.
  }
\[
\soph{c}(x) - \soph{c+e\log c}(x) \ge  \frac{3}{4}|x|. 
\]
\end{theorem}
\newcounter{sophUnstable}
\setcounter{sophUnstable}{\value{theorem}}

This theorem shows that for a fixed string ``the sophistication of this string'' corresponds to a {\em function} (of $c$),
rather than a single number (in a similar way as Kolmogorov introduced the closely related
structure function, see  section~\ref{sec:overview}).

\section{Related definitions of sophistication} \label{sec:overview}

We describe related notions of sophistication and present a few definitions.
No definitions or results from this section are needed in later sections.
For a recent overview paper, we refer to~\cite{ShenVereshchaginAlgMinStat2015}.
\medskip

The first approach to define some notion of sophistication 
goes back to Kolmogorov~\cite{KolmogorovTallinn,KolmogorovAbstract}
(see~\cite{CGG89}) and uses the definition of a typical string in a set.
\begin{definition}
  A string $x$ is $c$-{\em typical} in a finite set $S$ containing $x$
  iff \[C(x|S) \ge \log |S| - c.\] 
\end{definition}
For such $S$, a literal representation of the lexicographic index of $x$ in $S$ 
(of length $\log |S| + O(1)$) is almost a shortest description 
for $x$ given $S$.
By a counting argument, one can show that all but at most a fraction $2^{-c}$ 
of elements in a set are $c$-typical.\footnote{
  On the other hand, any set must have 
  non-typical elements unless the set contains a lot of 
  mutual information with the Halting problem~\cite{EpsteinLevin}.}
Kolmogorov asked whether there exist strings that are not typical in any
finite set with small Kolmogorov complexity\footnote{
  The Kolmogorov complexity of a set is the length of 
  a shortest program that prints all its elements and halts.
  }.

In~\cite{Shen83,Vyu87,GTV01} a positive answer was shown, {\em i.e.},
some strings are only typical in sets of complexity close to the length of the string.
Kolmogorov called such strings absolutely non-stochastic, because they 
have high mutual information with the Halting problem. 
It is believed that such strings can not appear with high probability in a statistical experiment. 
We define the {\em non-stochasticity} of a string as the minimal complexity of a set in which 
the string is $c$-typical:
\begin{definition}\label{def:nonstochasticity}
$
  \nstoch{c}(x) = \min \left\{ C(S) : x \text{ is $c$-typical in } S \right\}.
$
\end{definition}

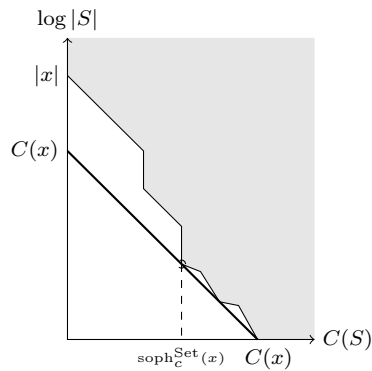
\begin{figure}
  \centering
  \begin{tikzpicture}[scale=0.5]
    \def\hght{8}
    \def\wdth{6.5}
    \newcommand{\someline}{(2,5) -- (2,4) -- (3,3) -- (3,2) -- (3.5,1.8) -- (4,1) -- (4.5, 0.9) -- (5,0)}
    \fill[gray!20] (0,\hght) -- (0,7) -- \someline -- (\wdth,0) -- (\wdth,\hght) -- cycle;  
    \draw[<->] (0,\hght) node[anchor=south] {$\log |S|$} -- (0,0) -- (\wdth,0) node[anchor=west] {$C(S)$} ;
    \draw (0,7) node[anchor=east] {$|x|$} -- \someline ;
    \draw[black,thick] (0,5) node[anchor=east,black] {\small{$C(x)$}} -- (5,0) node[anchor=north,black] {\quad \small{$C(x)$}};

   \draw[dashed] (3,2) circle (3pt) -- (3,0) node[anchor=north] {\tiny{$\soph{c}^{\text{Set}}(x)$}};
  \end{tikzpicture}
  \caption{
     The structure set of a string $x$ is the set of all pairs $(i,j)$ for which there exists 
     an $x$-containing set of complexity at most $i$ and cardinality at most $2^j$. Such a set is 
     schematically represented above in gray.
   }
  \label{fig:structureFunction}
\end{figure}

Kolmogorov also considered a more restrictive class of ``good'' set-models for a string $x$.
To understand this criterion, consider the  {\em structure set}, which is the set of 
all pairs $(i,j)$ for which there is an $x$-containing set 
of complexity at most $i$ and cardinality at most $2^j$, see Figure~\ref{fig:structureFunction}. 
Ignoring $O(\log |x|)$-terms, the set contains the points $(C(x),0)$ and $(0, |x|)$ 
witnessed by the set $\{x\}$ and the set of all strings of length $|x|$.
Note that if the set contains a point $(i,j)$, it also contains 
the points $(i+k, j-k)$ for all $k \le j$.\footnote{
  Partition the set in subsets of size at most $2^{j-k}$, this 
  increases the complexity of the $x$-containing set by at most~$k$.}
Hence, the lower border of the set, called structure function, decreases with at least slope one
(still ignoring $O(\log |x|)$ terms).
No point appears below the line $i+j = C(x)$, otherwise the corresponding set could be used to construct 
a program for $x$ of size less then $C(x)$.
Cover~\cite{Cover85} (see also~\cite[Sect. 14.12]{CoverElements} and~\cite[Sect 5.5.1]{livi}) 
mentioned explicitly the left-most place where the set approaches this line, 
which we call  {\em set sophistication} of~$x$: 
\begin{definition}\label{def:setsoph}
 $
\soph{c}^{\text{Set}}(x) = \min \{C(S) : x \in S \wedge C(S) + \log |S| \le C(x) + c\}.
$
\end{definition}


By the following theorem and lemma, sophistication, set-sophistication and non-stochasticity for a string $x$ are all $O(\log |x|)$-close.
\begin{theorem}[\cite{vv04}]\label{th:nonstochCloseToSetsoph}
  For all $x$, the functions $\nstoch{c}(x)$ and $\soph{c}^{\text{Set}}(x)$ are $O(\log |x|)$-close.
\end{theorem}

\begin{lemma}[\cite{Vit06}]\label{lem:setsophVsSoph}
  For all $x$, the functions $\soph{c}(x)$ and $\soph{c}^{\text{Set}}(x)$ are $O(\log |x|)$-close.
\end{lemma}

All these sophistication functions are unstable: increasing the parameter $c$ by $O(\log |x|)$, 
the function values can drop from maximal value $|x|-O(\log |x|)$ to $O(\log |x|)$.

\begin{corollary}[of Theorem IV.4 in \cite{vv04}]\label{cor:allshapes}
  There exists $e$ such that for all $c$ there exist infinitely many $x$ such that 
  \[
  \soph{c}^{\text{Set}}(x) \ge |x| - e\log |x| 
  \quad \text{and} \quad 
  \soph{c +e\log |x| }^{\text{Set}}(x) \le e\log |x| .\;\footnote{
Theorem IV.4 in~\cite{vv04} states that every decreasing function $f$ 
is $(C(f)+O(\log |x|))$-close to the function 
\[
\lambda_x(k) = \min \left\{ C(S) + \log |S|: S \ni x \wedge C(x) \le k \right\}
\]
of some string of length $f(0)$. 
(We use plain complexity in the definition of $\lambda_x$, because all results hold up
to $O(\log |x|)$ terms). 
For fixed $x$, the function $\lambda_x$ is the inverse of $\soph{x}^{\text{Set}}$.
}
  \]
\end{corollary}

\medskip

In~\cite{GTV01} it was shown that a sufficient set of almost minimal complexity of a string~$x$ can be computed from an initial segment of the binary code of the number of 
halting programs of length~$C(x)$. Hence, such a set contains high mutual information with the Halting problem for short programs
(see~\cite{ShenAlgorithmicStatistics}).
In~\cite{SSS,SSSdifferent} it is argued that in some cases this statistic can hardly be interpreted as a denoised version of~$x$. 
In fact, compared to $x$, a sufficient two part-code $(S,z)$ 
(where $z$ is the lexicographic index of $x$ in $S$) can contain different computational information from $x$, although $C(x|S,z)$ and 
$C(S,z|x)$ are both small. The proposed solution was to impose the existence of a 
computable bijection of small complexity between $x$ and $(S,z)$. 
This is equivalent to the requirement that there exists a short total program computing $S$ from $x$.
In~\cite{SSSdifferent,VereshchaginNewApproach} it was shown that this version
of sophistication can be much larger than set-sophistication. In fact, the
result shows that strings with large such sophistication can appear with
non-negligible probability in some statistical experiments.

\medskip
Until now, we considered two model types in the definitions of sophistication.
In Definition~\ref{AF}, we used computable functions 
that are $c$-{\em sufficient} for $x$, {\em i.e.}, functions $f$ for which a string $d$ exists
such that $f(d) = x$ and $C(f) + |d| \le C(x)+c$. In Definition~\ref{def:setsoph}, we considered
$c$-sufficient sets for $x$, {\em i.e.}, sets $S$ containing $x$ for which $C(S) + \log |S| \le C(x) + c$.
Another popular model type are
computable probability density functions~$P$. Such a $P$ is $c$-sufficient for $x$ if
$C(P) + \log (1/P(x)) \le C(x)+c$~\cite{GTV01}.\footnote{\label{fn:prefixplain}
  This probabilistic sufficiency criterion was defined in~\cite{GTV01} in terms of 
  prefix-free complexity, because $2^{-K (x|P)}$ defines a probability 
  distribution and hence, it is natural to compare it with $P(x)$. 
  Prefix complexity and plain complexity differ
  by at most $O(\log |x|)$~\cite{livi}, and this precision is sufficient for our discussion.}
In a similar way, probabilistic sophistication at 
significance level~$c$ is defined as the 
probability density function of minimal complexity that is $c$-sufficient.
By~\cite[Lemmas 7.1 and 7.2]{Vit06} all these variants of 
sophistication are $O(\log |x|)$-close.

In order to generalize the notion of sophistication for (infinite) sequences, 
Koppel~\cite{koppel87,koppel88,koppel91} considered monotone computable functions $f$ as models. 
The sufficiency criterion for the two-stage code for $x$ is the existence 
of a string $d$ such that $f(d)=x$ and $Km(f) + |d| \le H(x)+c$ where $H(x)$ is the 
minimal length of a two-part description for $x$ on some special monotone machine and $Km(x)$
denotes the monotone Kolmogorov complexity relatively to the machine considered.
It is not hard to show that $H(x) = C(x) + O(\log |x|)$ and that this notion of sophistication is $O(\log |x|)$-close to 
the aforementioned notions.\footnote{
 It is unclear whether $H(x) = Km(x) + O(1)$.
 }
On this model, Koppel defined sophistication and depth for sequences in two variants, 
and for each variant he showed that sophistication and depth are equal up to constants.


\bigskip

The last variant of sophistication is effective complexity~\cite{GML96,GML04}. 
This notion uses a probability density function $P$. 
Inspired by an information-theoretic solution of the problem of Maxwell's Demon, total entropy of
$P$ has been defined as $C(P) + H(P)$, where $H(P) = \sum_x P(x) \log_2 (1/P(x))$ denotes the Shannon
entropy of $P$.\footnote{The definition of total entropy used in~\cite{GML96,GML04} is $K(P) +
H(P)$. Notice that plain and prefix complexity are close ($|K(P) - C(P)| \le O(\log
C(P)$). See also footnote~\ref{fn:prefixplain}.} A probability density function $P$ is a $c$-{\em good}
model for $x$ if
$C(P) + H(P) \le C(x)+c$ and $\log (1/P(x)) \le H(P)+c$.\footnote{
  In fact, in~\cite{GML96} the precision 
  for which these inequalities should hold is not discussed. 
  Also, the authors suggest that the computation time of 
  a program for $P$ is bounded by some computable function.
  In~\cite{effCompVsDepth}
  the first requirement $c = \delta|x|$ is chosen for some $\delta>0$ and in the second 
  requirement a different parameter is chosen. 
  Furthermore, $P$ should be computable as a real function 
  and no restrictions on the computation time are considered. 
  Also, $K(P)$ was replaced by $K(P, H(P))$.} 
The $c$-effective complexity is the minimal complexity of a $c$-good model.
In~\cite[Lemma 21]{effCompVsDepth}, it is shown that effective complexity is 
$O(\log |x|)$-close to set-sophistication.\footnote{
  Indeed, if $P$ is $c$-good then it is $(2c)$-sufficient. 
  For the other direction, note that at most $2^{H(P)+c + 1}$ elements satisfy 
  $\log (1/P(x)) \le \lceil H(P) \rceil + c$,
  and these elements can be computed given $P$ and $\lceil H(x)\rceil \le C(x) + c \le |x|+c+O(1)$.
  Hence a $c$-good model defines a $(c+O(\log |x|))$-sufficient set.}
In~\cite[Theorem 18]{effCompVsDepth} it was also shown that strings with 
high effective complexity have high computational depth. Moreover, 
the proof shows that effective complexity is upper bounded by the Busy Beaver
logical depth with slightly bigger significance.
Our Theorem~\ref{th:ldepthVsSoph} implies also the other direction,  {\em i.e.}, 
that effective complexity is $O(\log |x|)$-close to Busy Beaver 
logical depth.

\section{Sophistication and Busy Beaver logical depth are close}\label{sec:relacoes}

Koppel~\cite{koppel87} proved an equivalence between logical depth and sophistication for infinite
sequences. 
For such sequences, and for fixed significance, depth is defined as the minimal complexity of a total function rather than the
minimal complexity of an upper bound for a number. 
In this section we show that sophistication and Busy Beaver logical depth of a string $x$ are $O(\log |x|)$-close functions.

\setcounter{oldTheorem}{\value{theorem}}
\setcounter{theorem}{\value{ldepthVsSoph}}
\addtocounter{theorem}{-1}
\begin{theorem}
  The functions $\ld{c}^{bb}(x)$ and $\soph{c}(x)$ are $O(\log |x|)$-close, {\em i.e.},  for some $e$ and
  for all $c$ and $x$ with $|x| \ge e$:
  \begin{eqnarray*}
     \ld{c + e\log |x|}^{bb}(x) &\le& \soph{c}(x)+e\log |x| \\
     \soph{c + e\log |x|}(x) &\le& \ld{c}^{bb}(x) + e\log |x|.
  \end{eqnarray*}
\end{theorem}
\setcounter{theorem}{\value{oldTheorem}}

In the Appendix we provide an alternative and more technical proof of this
result involving Chaitin $\Omega$-numbers that might be of interest for people
with background in the theory of algorithmic randomness.
\begin{proof} 
  To prove the first inequality, we assume $c \le |x| + O(1)$; otherwise the theorem
  follows directly. Consider $p$ and $d$ such that 
  \begin{enumerate}
    \item[$A1$.] the function $U(p,\cdot)$ is total,
    \item[$A2$.] $U(p,d) = x$,
    \item[$A3$.] $|p| + |d| \le C(x) + c$.
  \end{enumerate}
  For later use, note that by assumption on $c$ we have that $|p|$ and $|d|$ are bounded by $2|x| + O(1)$.
  We need to construct $q$ and $r$ such that
  \begin{enumerate}
    \item $U(q) = x$ and $|q| \le |p| + |d| + O(\log |x|)$, \label{cond:q}
    \item $\halttime(q) \le \halttime(r)$, \label{cond:time}
    \item $|r| \le |p| + c + O(\log |x|)$. \label{cond:length}
  \end{enumerate}

  \smallskip
  The idea of the construction is to let $r$ be a shortest program for the maximal 
  computation time needed to simulate $U(p,e)$ for some $e$ of length $|d|$. 
  Let us define this quantity more formally.

  {\em Construction of $q$.}
  For a string $y$, let $\ovl{y}$ be a computable prefix-free encoding of length $|y| +
  2\log |y|$. (For example $\ovl y = b_10b_20\dots b_{\log |y|}1y$ where $b$ is $|y|$ in binary.) 
  Let $V$ be a machine such that $V(\ovl{y}e) = U(y,e)$ if
  $U(y,e)$ is defined. Thus $U(w\ovl ye) = V(\ovl{y}e) = U(y,e)$ for some $w$ and all $y,e$. 
  Let $q = w\ovl{p}d$. Thus, $U(w\ovl pd) = U(p,d) = x$.
  Recall that $|p| \le 2|x| + O(1)$, hence,
  $|q|$ satisfies condition \ref{cond:q}:
  \[
   |q| \le O(1) + \left(|p| + O(\log |p|)\right) + |d| \le |p| + |d| + O(\log |x|)\,.
   \]

   {\em Construction of $r$.} 
  Let 
    $$
    \ds{t=\max_e \left\{ \halttime(w\ovl{p}e) : |e| = |d| \right\}}.
    $$ 
  The program $r$ is a shortest program printing a string containing $t$ zeros. Clearly, the
  running time of this program is at least~$t$ and by construction this exceeds $\halttime (q)\ge 
  \halttime (w\ovl pd)$, which verifies condition~\ref{cond:time}. For condition~\ref{cond:length}
  notice that to compute $t$, we only need to know $p$ and $|d| \le 2|x| + O(1)$, hence,
  \[
   |r| \le |\ovl p| + O(\log |d|) \le |p| + O(\log |x|).
   \]
   This concludes the proof of the first inequality.

  \medskip

  Now we prove the second inequality. For each $k,l$ such that $l \le k$ consider a sequence of strings and markers 
    $$ \ds{x_1, x_2, \dots, x_i, \Box, x_{i+1}, \dots, x_j, \Box, x_{j+1}, \dots} 
   $$ 
  that can be enumerated as follows: dovetail all programs of length $l$ and $k$, and enumerate the output of the
  $k$-bit programs in order of computation time. Each time a program of length $l$ halts,
  also append a marker to the series (if $k$-bit programs with the same computation time appear,
  append the marker last).
    One easily observes that:
    \begin{enumerate}
      \item\label{en:uniform} the sequence can be enumerated from $k,l$,
      \item \label{en:markers} there are at most $2^l$ markers, and at most $2^k$ strings, 
      \item if a program of length $k$ outputs $x$ in at most $BB(l)  = \max\left\{ \halttime(p) : |p| \le l \right\}$ steps, 
	then $x$ appears in the sequence before its last marker.
    \end{enumerate}

  \noindent
  The second inequality of the theorem follows from the following lemma.
  \begin{lemma}\label{claim}
  Every string that appears before the last marker in a sequence satisfying properties
  \ref{en:uniform} and \ref{en:markers} above, satisfies
  \[
    \soph{k - C(x) + O(\log k)}(x) \le l + O(\log k).
    \] 
  \end{lemma}
  We show that this lemma implies the inequality. 
  Assume $c \le |x| + O(1)$, otherwise the inequality holds for trivial reasons.
  Let $l = \ld{c}^{bb}(x)$. 
  There is a program $p$ for $x$ with $|p| \le C(x)+c$ 
  that runs in at most $BB(l)$ steps (and in more than $BB(l-1)$ steps). 
  Let $k = |p|$, thus $l \le k$ 
  (if $|p| < l$, the running time of $p$ would be at most $BB(l-1)$). 
  Enumerate a sequence as described above with parameters $l$ and $k$. 
  Notice that $x$ appears in the sequence before the last marker.  By the claim, we have
  \[
      \soph{c + O(\log k)}(x) \le l + O(\log k).
  \]
  The inequality follows from this and $k \le C(x) + c \le 2|x| + O(1)$.
%
\qed
\end{proof}
To complete the proof of Theorem \ref{th:ldepthVsSoph} we prove Lemma~\ref{claim}. 

\begin{proof}[of Lemma \ref{claim}]
  For any computable function $f$, let $C(f)$ denote the minimal length of a program that computes $f$.
  For any $x$ as in the Lemma,  we need to show that there is a computable function $f$ such that:
  \begin{enumerate}
    \item $C(f) \le l + O(\log k)$          \label{cond:claim_l}
    \item $C(f) + |d| \le k + O(\log k)$ for some $d \in f^{-1}(x)$.   \label{cond:claim_kd}
  \end{enumerate}

  Consider a segment of strings $x_{i+1}, \dots, x_j$ in the sequence, separated by two markers
  $\Box$ that contains $x$. We associate a function $f$ to this segment that maps 
  the lexicographic first $j-i$ strings to $x_{i+1},
  \dots, x_j$ and all other strings to the empty one.  Notice that  $f$ is computable, and can be
  computed from $k,l$ and the number of markers that precede the defining segment (which is at most $2^l$). 
  This implies $C(f) \le l + O(\log kl) = l + O(\log k)$, i.e., condition~\ref{cond:claim_l}.

  It remains to show condition~\ref{cond:claim_kd}.
  Let $\delta = \log (j-i)$, {\em i.e.} the logarithm of the size of the segment. 
  Observe that at most
  $2^{k-\delta}$ segments in the sequence have size at least $2^\delta$ (by assumption \ref{en:markers}).
  Hence, $C(f) \le k - \delta + O(\log kl\delta)$. 
  Since the segment contains $x$, there is a $d$ such that $f(d)=x$, 
  and by construction $|d| \le \delta$.  Hence 
  $C(f) + |d| \le (k - \delta) + \delta + O(\log (kl\delta))$, i.e. condition~\ref{cond:claim_kd}.
%
%
%
%
  \qed
\end{proof}

\section{Sophistication is unstable}\label{sec:sophunstable}

In~\cite{soph} the authors conjectured that Koppel's definition of sophistication might not be
stable, in the sense that small changes in the significance $c$ level could
drastically change the value of $\soph{c}(x)$.  To avoid this problem, they proposed a different
sophistication measure where they incorporated the significance level as a penalty in the formula
obtaining a robust measure, called coarse sophistication.  However, one can argue that this measure 
is not robust in the sense that drastic changes can happen for slight
changes of the weight of the penalty function~\cite{SSS}.

In Section~\ref{sec:overview}, we used \cite[Theorem IV.4]{vv04} to show that (most variants of)
sophistication are unstable if the significance is increased by $O(\log |x|)$. 
With the same proof technique, one can show that also
sophistication when defined with prefix complexity is unstable with constant changes of the
significance.

One might ask whether sophistication functions defined with plain complexity are also unstable with
constant changes in the precision? We provide a positive answer to this question.


\setcounter{oldTheorem}{\value{theorem}}
\setcounter{theorem}{\value{ldepthVsSoph}}
\addtocounter{theorem}{-1}
\begin{theorem}\label{th:sophUnstable}
For some $e$ and for large $c$ there are infinitely many $x$ such that
\[
\soph{c}(x) - \soph{c+e\log c}(x) \ge  \frac{3}{4}|x|. 
\]
\end{theorem}
\setcounter{theorem}{\value{oldTheorem}}

The proof also uses a technique inspired by the proof of~\cite[Theorem
IV.4]{vv04}. However, some technical difficulties appear because we are using
plain machines. Let us explain the problem.  In the definition of
sophistication of $x$, we consider pairs of strings $(p,d)$ such that $U(p,d) =
x$. For some $k$ there are $2^k$ strings of length $k$, but there are
$(k+1)2^k$ pairs $(p,d)$ with $|p|+|d|=k$. In~\cite{vv04} self-delimiting 
programs are used and the combinatorial part of the argument uses that the
amount of two-part codes of length $k$ is at most $2^k$.
In this paper we do not use self-delimiting machines, 
and therefore the combinatorial argument needs a bit more care.

\begin{lemma}\label{lem:twopartShorter}
 For some $c'$ and for all $k$ and $x$ such that $k + \log k \leq |x|$ we have
  \[
   \soph{|x|-C(x)-\log k + c'}(x) \le k + c'.
  \]
\end{lemma}

Recall that sophistication is defined for negative significance. 
This lemma even proves that sophistication can be negative for all random 
strings, i.e., strings $x$ for which $C(x) \ge |x|$.

\begin{proof}
Let $n = |x|$. In order to prove the lemma it is sufficient to show that there is a two-part
description $(p,d)$ for $x$ satisfying $|p|+|d| \leq n - \log k + O(1)$ and $|p| \leq k + O(1)$. The
idea to prove it is to use the length of $|p|$ to encode the last $\log k-1$~bits of~$x$.

Let~$i$ be the index of the last $\log k-1$~bits of $x$ in the lexicographic order of strings; 
({\em i.e.}, $x_{n - \log k + 2}\dots x_n$ is the $i$-th string in the sequence $\varepsilon, 0, 1, 00,
01, \dots$). Notice that $i < k$.

 Let $p$ be the program that on input $d$ first prints $x_1\dots x_i$, subsequently prints $d$, and
 finally prints $x_{n-\log k+2}\dots x_n$. Clearly, the above description defines a total function.
 Moreover, only the information in $x_1\dots x_i$ is needed to evaluate this function, since the
 last part of the output can be computed from $i$.
Hence, we can construct $p$ such that $|p| = i + O(1) \leq k + O(1)$.

Furthermore, for $d = x_{i+1}\dots x_{n-\log k+1}$ we have $U(p,d) = x$ and 
  $|p| + |d| \leq (i + O(1)) + (n - i - \log k) \le n - \log k + O(1)$.  
  \qed
\end{proof}

\begin{proof}[of Theorem~\ref{th:sophUnstable}.]
  It is sufficient to show that for all $k,c$ there is a string $x$ of length $k + \log k + 2$  such
  that $\soph{c-O(\log c)}(x) \ge k$ and $\soph{c + O(1)}(x) \le k/8 + O(1)$.
   
  Our construction of $x$ implies that $C(x) \ge k-c$. 
  Hence, applying Lemma~\ref{lem:twopartShorter} with $k \leftarrow k/8$ implies $\soph{c + O(1)}(x)
  \le k/8 + O(1)$; indeed, the significance is
  \[ 
    |x| - C(x) - \log (k/8) + c' \le (k + \log k + 2) - (k - c) - \log k + 3 + c' = c + O(1).
  \]
  The inequality $\soph{c-O(\log c)}(x) \ge k-c$ follows by the requirements that
  $C(x) \le k-c + O(\log c)$ and that there exist no pairs $(p,d)$ such that 
  \begin{enumerate}
     \item $U(p,d) = x$ and $|p|+|d| < k$, \label{en:condX}
     \item $|p| < k-c$ and $U(p,y)$ is defined for all~$y$ 
       such that $|p|+|y| < k$.\label{en:condProg}
   \end{enumerate}

  \smallskip
  Let us summerize the properties needed in the construction of $x$ (of length $k + \log k + 2$).
  The complexity should be
  \[
   k - c \le C(x) \le k - c + O(\log c),
  \]
  and there should not exists pairs $(p,d)$ satisfying conditions~\ref{en:condX} and~\ref{en:condProg} above.
  
  \medskip

  {\em Construction of $x$.} 
  We keep a list of all strings of length $k + \log k + 2$.  At each stage we mark some strings and 
  the lexicographic first string without a mark is the current candidate for $x$.
  At each stage, marks are given as follows: we dovetail all programs $p$, and if a program of length less
  than $k-c$ halts with an output in the list, then that output is marked. Clearly, there are less
  than $2^{k-c}$ strings that are marked in this way. Secondly, if a program $p$ is found satisfying
  condition \ref{en:condProg}, {\em i.e.}, for which the computations $U(p,y)$ terminate for all $y$ such
  that $|y| + |p|<k$, then all strings $U(p,y)$ in the list are
  simultaneously marked. These marks appear in less than $2^{k-c}$ different stages, and the
  number of such marks is less than $\sum_{i=0}^k 2^i2^{k-i} < (k+1)2^{k}$.  Hence, the total number
  of marked strings is less than $(k+1)2^{k+1} \le 2^n$ which means there is always a candidate for~$x$ 
  and at some stage the new candidate remains permanent.  By construction, $C(x) \ge k-c$ and
  there is no pair $(p,d)$ for which both conditions \ref{en:condX} and \ref{en:condProg} are
  satisfied.

  Now we have to prove that $C(x) \le k - c + O(\log c)$.  $x$ can be computed from $k,c$ and 
  the total number $N$ of replacements of the candidate for $x$.
  Since there are  less than $2^{k-c} + 2^{k-c}$ stages where new marks are given, we have
  $N < 2^{k-c+1}$ times and hence $C(x|k,c) \le k-c + O(1)$.  
  In fact, if $N$ is represented in binary with $k-c$ bits, we can compute $k$ from $c$ and 
  the length of this representation. Thus $C(x|c) \le k-c + O(1)$ and hence $C(x) \le k-c + O(\log c)$.  
  \qed
\end{proof}

\section{Sophistication and Busy Beaver logical depth are not $O(1)$-close}
\label{sec:ldepthDiffersSoph}

In this section we investigate whether there exists an $O(1)$-close
relation between sophistication and logical depth.
More precisely, for every $c$ can we find an $e$ such that 
\[
  \soph{c+e}(x) \le \ld{c}^{bb}(x) + e \quad \text{  and  } \quad \ld{c+e}^{bb}(x) \le \soph{c}(x) + e?
\]
The following theorem provides a negative answer:

\begin{theorem}\label{th:ldepthDiffersSoph}
For all large $l$ there exist infinitely many strings $x$ such that
$\ld{l}^{bb}(x) \geq |x| - O(l)$ and $\soph{0}(x)\leq O(l^2 2^l)$.
\end{theorem}



We explain informally why an equivalence with $O(1)$ precision fails. 
In the definition of sophistication of $x$, we consider pairs of strings $(p,d)$
such that $U(p,d) = x$. As noted before, for all $k$ there are $2^k$ strings of
length $k$, but there are $(k+1)2^k$ pairs $(p,d)$ with $|p|+|d|=k$. This
suggests that strings might exist that have a two-part code $(p,d)$ for which
$|p|+|d|$ is smaller than~$C(x)$.  In other words, this suggests that
sophistication can be finite even for negative significance. 
For an explicit example, choose a string $x$ for which $C(x) \ge |x|$ and apply 
Lemma~\ref{lem:twopartShorter} in  Section~\ref{sec:sophunstable}. 
For all random and almost random strings, sophistication with negative
significance can still be small.  There exist $x$ that are only compressible by
a small amount and for which the logical depth is high for small significance. 
Such $x$ are almost random, and hence can have small sophistication even with
negative significance.

We now prove Theorem~\ref{th:ldepthDiffersSoph} by combining 
Lemma~\ref{lem:twopartShorter} in  Section~\ref{sec:sophunstable} with the following lemma.

\begin{lemma}\label{lem:controlLdepthAndC}
For some $c$, for all $d$ and for all $n>d$ there is $x$ of length $n$ such that
  \[\begin{array}{c}
     C(x) \geq n - d \,, \\[3mm]
     \ld{d - 2\log d - c}^{bb}(x) \geq n - d \,.
  \end{array}\]
\end{lemma}

\begin{proof}
We prove the existence of such strings for $n>d$, since for the other case is trivial. 

Let $x$ be the lexicographically first string of length $n$  which  is incompressible in time $BB(n - d)$, {\em i.e.} there is no program strictly shorter than $n$ computes $x$ in $BB(n - d)$ steps.

To show the inequalities in the statement of the lemma, it is sufficient to show that
 \[
   n - d < C(x) \leq n - d + 2\log d + O(1).
 \]
 For the right inequality, notice that we can compute $x$ from $BB(n -d)$ and $n$. Furthermore, with $O(1)$ bits of information, $n$ can be computed from $d$ and the length of a witnessing program for $BB(n - d)$ (notice that a program witnessing $BB(n-d)$ has length $n-d + O(1)$). 
 Hence $x$ has a program of length $n - d + 2\log d + O(1)$.
 
For the left inequality, notice that by the right inequality we have $C(x) < n$ for large~$d$. 
By choice of $x$, any program
producing $x$ of length at most $n-1$ must do it in time longer than $BB(n-d)$, and by definition of
$BB(n-d)$ this program must be strictly longer than $n-d$.  
\qed
\end{proof}

\begin{proof}[of Theorem \ref{th:ldepthVsSoph}]
  Let $c'$ be the constant from Lemma \ref{lem:twopartShorter}.
  For any large $k$ we apply  Lemma \ref{lem:controlLdepthAndC} with $d = \log k - c'$ to obtain a
  string $x$ of complexity $C(x) \ge |x|-\log k + c'$. Apply this bound to Lemma
  \ref{lem:twopartShorter}; the significance of the sophistication is at most  $|x| - (|x|-\log k + c')
  - \log k + c' = 0$ and  we conclude that $\soph{0}(x) \le k+c' \le O(2^d)$.

  At the same time $x$ satisfies $\ld{d - 2\log d -c}(x) \ge |x|-d$. Hence setting  $l = d -
  2\log d -c$ the equations of the Theorem \ref{th:ldepthVsSoph} are satisfied. 
  Since $k$ can be any large number, also $d$ and $l$ can be any large number, completing the proof.
  \qed
\end{proof}

\bigskip

If sophistication can be finite for negative significance, it would be 
fair to compare $\ld{O(1)}(x)$ to $\soph{\ell}^{bb}(x)$ where $\ell$ equals the minimal
value of the significance for which sophistication is finite.  
This value is $-\log C(x) + O(1)$ for every~$x$.
The following lemma implies that even with such a correction we can not have a 
correspondence with sublogarithmic terms in the significance.

\begin{lemma}\label{lem:negativeSignificance}
  There exists an $e$ such that for all $x$ and $c \ge 0$: $\soph{-2c-e}(x) \ge 2^c$. 
\end{lemma}

If $f$ is a sublogarithmic function, this lemma implies that $\soph{-\log |x| + f(|x|)}(x)$ is
at least proportional to  $\sqrt{|x|}$ for large $x$. (And is finite for $x$ such that $C(x) \ge |x|$.)
On the other hand, $\ld{O(1)}(x) \le bb(|x|+O(1))$ for all random $x$.
Hence, this approach does also not provide a close correspondence between depth and
sophistication.

\begin{proof}[of  Lemma~\ref{lem:negativeSignificance}]
  Let $e$ be a large enough constant that will be determined later.
  Suppose that $\soph{-2c-e}(x) < 2^c$ for some $x$ and $c \ge 0$. Let $p$ and
  $d$ be such that $U(p,d) = x$ with $|p| < 2^c$ and \[|pd| \le C(x) - 2c - e.\]
  Let $\ovl{p}$ be a self-delimiting encoding of $p$ of length at most 
  $2\log |p| + |p| \le 2c + |p|$. This code can be concatenated to $d$ to get a
  program for $x$ and this implies that $C(x) < 2c + |p| + |d| + e$ for some
  large enough $e$. By assumption on $|pd|$ this implies $C(x) < C(x)$, a
  contradiction. \qed 
\end{proof}

To study the relationship between depth and sophistication with more precision,
one can avoid the pathology of two part codes by using self-delimiting programs for
the total functions. 
Such programs can be concatenated with an argument without blank between both strings. 
This implies that one also needs to use self-delimiting programs 
for $x$, or otherwise again pathological examples can be constructed.  
More formally, one uses prefix-free Turing machines, which are
machines for which the set of halting programs is a prefix-free set.
There exists a universal such machine and we denote Kolmogorov complexity, sophistication and Busy
Beaver logical depth of $x$ on such a machine as $K(x)$, $\soph{c}^K(x)$ and $\ld{c}^K(x)$.  
It was shown~\cite[Theorem 3.2.2]{bruno} that with these definitions
sophistication and logical depth are still not $O(1)$-close.\footnote{
  The formulation in \cite[Theorem 3.2.2]{bruno} uses $I(x;H) = K(x) - K^H(x)$ with $K^H(x)$
  the Kolmogorov complexity on a machine that has an oracle for the Halting problem. 
  To obtain Theorem~\ref{th:phdBauwens} from this, use the 
  folklore result: $\ld{0}(x) \ge I(x;H) + O(\log I(x;H))$.
  }

\begin{theorem}[\cite{bruno}]\label{th:phdBauwens}
 For all $c$ and $e$ there exist infinitely many $x$ such that
 \[
 \soph{c}^K(x) \ge \left(\ld{0}^K(x)\right)^e.
 \]
 For all $c$ there exist $\varepsilon > 0$ and infinitely many $x$ such that
 \[
 \soph{c}^K(x) \ge \varepsilon|x| + \ld{0}^K(x).
 \]
\end{theorem}

\section*{Acknowledgments}
The authors are grateful to the anonymous reviewers and to
to Alexander Shen and 
Nikolay Vereshchagin for useful comments and discussions.\\
This work was partially supported by the national science foundation: Funda\c{c}\~{a}o para a Ci\^{e}ncia e Tecnologia, through the scholarships SFRH/BPD/76231/2011,
SFRH/BPD/75129/2010 and SFRH/BD/33234/2007, through the project $CSI^2$
with the reference PTDC/EIA-CCO/099951/2008 and grants of Instituto de Telecomunica\c{c}\~oes. The second author was also supported by NAFIT ANR-08-EMER-008-01 project.


\bibliographystyle{spmpsci}
\bibliography{bibtex,kolmogorov}

\begin{thebibliography}{10}
\providecommand{\url}[1]{{#1}}
\providecommand{\urlprefix}{URL }
\expandafter\ifx\csname urlstyle\endcsname\relax
  \providecommand{\doi}[1]{DOI~\discretionary{}{}{}#1}\else
  \providecommand{\doi}{DOI~\discretionary{}{}{}\begingroup
  \urlstyle{rm}\Url}\fi

\bibitem{soph}
Antunes, L., Fortnow, L.: Sophistication revisited.
\newblock Theory of Computing Systems \textbf{45}(1), 150--161,
  Springer--Verlag (2009)

\bibitem{afmv06}
Antunes, L., Fortnow, L., van Melkebeek, D., Vinodchandran, N.: Computational
  depth: concept and applications.
\newblock Theoretical Computer Science \textbf{354}(3), 391--404 (2006)

\bibitem{bruno}
Bauwens, B.: Computability in statistical hypotheses testing, and
  characterizations of independence and directed influences in time series
  using kolmogorov complexity.
\newblock Ph.D. thesis, Ugent (2010)

\bibitem{ben88}
Bennett, C.: Logical depth and physical complexity.
\newblock In: A half-century survey on The Universal Turing Machine, pp.
  227--257. Oxford University Press, Inc., New York, NY, USA (1988)

\bibitem{Cha66}
Chaitin, G.: On the length of programs for computing finite binary sequences.
\newblock Journal of ACM \textbf{13}(4), 547--569, ACM Press (1966)

\bibitem{CGG89}
Cover, T., Gacs, P., Gray, R.: {K}olmogorov's contributions to information
  theory and algorithmic complexity.
\newblock The Annals of Probability \textbf{17}(3), 840--865 (1989)

\bibitem{CoverElements}
Cover, T., Joy, T.: Elements of Information Theory.
\newblock John Wiley \& sons (1991)

\bibitem{Cover85}
Cover, T.M.: The Impact of Processing Techniques on Communications., chap.
  {K}olmogorov Complexity, Data Compression and Inference., pp. 23--33.
\newblock J. Skwyrzynski, Martinus Nijhoff Publishers (1985)

\bibitem{EpsteinLevin}
Epstein, S., Levin, L.: On sets of high complexity strings.
\newblock CoRR \textbf{abs/1107.1458} (2011)

\bibitem{GTV01}
G{\'a}cs, P., Tromp, J., Vit{\'a}nyi, P.M.B.: Algorithmic statistics.
\newblock IEEE Transactions on Information Theory \textbf{47}(6), 2443--2463
  (2001)

\bibitem{GML96}
Gell-Mann, M., Lloyd, S.: Information measures, effective complexity, and total
  information.
\newblock Complexity \textbf{2}(1), 44--52 (1998)

\bibitem{GML04}
Gell-Mann, M., Lloyd, S.: Effective complexity.
\newblock Nonextensive entropy pp. 387--398 (2004)

\bibitem{DepthAndReducibility}
Juedes, D.W., Lathrop, J.I., Lutz, J.H.: Computational depth and reducibility.
\newblock Theoretical Computer Science \textbf{132}, 37--70 (1994)

\bibitem{kol65}
Kolmogorov, A.: Three approaches to the quantitative definition of information.
\newblock Problems of Information Transmission \textbf{1}(1), 1--7,
  Springer--Verlag (1965)

\bibitem{KolmogorovTallinn}
Kolmogorov, A.: Talk in information theory symposium.
\newblock In: Tallinn, Estonia (1973)

\bibitem{KolmogorovAbstract}
Kolmogorov, A.: Complexity of algorithms and objective definition of randomness
  \textbf{29}(4), 155 (1974).
\newblock Russian abstract of talk at Moscow Math. Soc. meeting. 4/16/1974.
  English abstract in \cite[p.438]{livi}.

\bibitem{koppel87}
Koppel, M.: Complexity, depth, and sophistication.
\newblock Complex Systems \textbf{1}, 1087--1091 (1987)

\bibitem{koppel88}
Koppel, M.: Structure.
\newblock In: R.Herken (ed.) The Universal Turing Machine: a Half-Century
  Survey, 2nd edition, pp. 403--419. Springer-Verlag (1995)

\bibitem{koppel91}
Koppel, M., Atlan, H.: An almost machine-independent theory of program-length
  complexity, sophistication, and induction.
\newblock Information Scinces \textbf{56}(1-3), 23--33, Elsevier Science
  Publishers Ltd. (1991)

\bibitem{LiVi97}
Li, M., Vit\'{a}nyi, P.: An Introduction to Kolmogorov Complexity and Its
  Applications.
\newblock Springer-Verlag (1997)

\bibitem{livi}
Li, M., Vit\'{a}nyi, P.: An Introduction to Kolmogorov Complexity and Its
  Applications.
\newblock Springer-Verlag (2008)

\bibitem{effCompVsDepth}
Nihat, A., Muller, M., Szkola, A.: Effective complexity and its relation to
  logical depth.
\newblock IEEE Transactions on Information Theory \textbf{56}(9), 4593--4607
  (2010)

\bibitem{ShenAlgorithmicStatistics}
Shen, A.: Algorithmic statistics: main results (2013).
\newblock In preparation

\bibitem{Shen83}
Shen, A.K.: The concept of (alpha, beta)-stochasticity in the {K}olmogorov
  sense and its properties.
\newblock Soviet Mathematics Doklady \textbf{28}(1), 295--299 (1983)

\bibitem{Sol64}
Solomonoff, R.: A formal theory of inductive inference, \textnormal{Part I}.
\newblock Information and Control \textbf{7}(1), 1--22, Academic Press Inc.
  (1964)

\bibitem{SSS}
Vereshchagin, N.: Algorithmic minimal sufficient statistic revisited.
\newblock Mathematical Theory and Computational Practice pp. 478--487 (2009)

\bibitem{VereshchaginNewApproach}
Vereshchagin, N.: Algorithmic minimal sufficient statistics: a new approach.
\newblock Theory of Computing Systems pp. 1--19 (2015)

\bibitem{ShenVereshchaginAlgMinStat2015}
Vereshchagin, N., Shen, A.: Algorithmic statistics revisited.
\newblock arXiv preprint arXiv:1504.04950  (2015)

\bibitem{vv04}
Vereshchagin, N., Vitanyi, P.: Kolmogorov's structure functions and model
  selection.
\newblock IEEE Transactions on Information Theory \textbf{50}, 3265--3290,
  Computer Society (2004)

\bibitem{SSSdifferent}
Verschagin, N.: On algorithmic strong sufficient statistics.
\newblock In: Proceedings of Computability in Europe (2013)

\bibitem{Vit06}
Vit{\'a}nyi, P.: Meaningful information.
\newblock IEEE Transactions on Information Theory \textbf{52}(10), 4617--4626
  (2006)

\bibitem{Vyu87}
V’yugin, V.V.: On the defect of randomness of a finite object with respect to
  measures with given complexity bounds.
\newblock Theory Prob. Appl. \textbf{32}(3), 508--512 (1987)

\end{thebibliography}

\appendix
\section{Machine invariance of logical depth}
\label{sec:bbdepthInvariant}

\begin{lemma}\label{lem:depthIsInvariant}
  For all universal
  Turing machines $U$ and $V$,
  there exist a constant~$c'$ 
  such that for all $c$ and $x$: $\ld{c,U}(x) \ge |x|$ [no Busy Beaver here!] implies
  \[\ld{c+c',V}^{bb}(x) \le \ld{c,U}^{bb}(x) + c'.\]
\end{lemma}

Note that for some universal machines there exist a
string $w$ such that $U(wx) = x$ for all $x$ and the computation requires at most $O(1)$ steps.
For such machines $U$ we have $\ld{|wx|}(x) \le O(1)$ and hence $\ld{|wx|}^{bb}(x) \le O(1)$. 
Other universal machines always scan the
input, and on such machines we have $\ld{|wx|}^{bb}(x) \ge bb(|x|) - O(1)$ for all $x$.
Hence, the assumption in the lemma is necessary.

\begin{proof}
  Let $w_V$ be the prefix such that  $V(w_Vp)$ simulates $U(p)$ for all~$p$. 
  Our result would follow easily if we assume that for any halting programs $p,q$ on $U$ such that
  $\halttime(p) \le \halttime(q)$ we have $\halttime(w_Vp) \le \halttime(w_Vq)$ on~$V$;  {\em i.e.}
  simulating $U$ on $V$ preserves the order of computation time.  
  Indeed, any pair~$(p,q)$ usable in the definition of depth on $U$ defines a pair~$(w_Vp, w_Vq)$ that can be used in the definition of depth on~$V$. 
  The program $w_Vp$ is minimal on~$V$ within $c + |w_V| + |w_U|$ error 
  (where $w_U$ is the string that allows to simulate~$U$ on~$V$).
  Hence, the pair~$(w_Vp, w_Vq)$ witnesses an increase of sophistication by at most~$|w_V|$ 
  for an increase of the significance of at most $c + |w_V| + |w_U|$. 

  In the case where the assumption is not true, we need to find $c'$ and a program of length at most
  $|q|+c'$ on $V$ that computes longer than~$\halttime(w_Vp)$ (where $c'$ does not depend on
  $p,q,c$). 
  Consider the following algorithm on input $q$: determine all programs $p$ that have running time at most
  $\halttime(q)$ on~$U$,  determine for all these $p$'s the maximal running time
  $T$ of a program $w_Vp$ on $V$ (assume for now that for finite $\halttime(q)$ there are 
  finitely many such $p$), and finally print a string of length~$T$.  For $(p,q)$ 
  usable in the definition of $\ld{c,U}(x)$, the algorithm with input $q$
  produces an output longer than $\halttime(w_Vp)$, and by universality there is a program
  of length $|q|+c'$ on~$V$ that prints this string and hence computes longer than $T$. 

  Above, we have assumed that only finitely many programs on $U$ have 
  a halting time at most $\halttime(q)$ for halting $q$. This assumption is not true in general, but 
  by the additional assumption of the lemma: $\ld{c,U}(x) \ge |x|$, it suffices to consider only 
  a finite subset of candidates: we only need the pairs $(p,q)$ on $U$ such that
  $|p| \le |x| + O(1)$ and $|x| \le \halttime(q)$, which implies $|p| \le \halttime(q) + O(1)$.
  The proof finishes by modifying the above algorithm such that it only considers programs~$p$ 
  for which $|p| \le \halttime(q) + O(1)$.
  \qed
\end{proof}

Recall that Bennett's definition of logical depth is the minimal computation time of a 
program on a prefix-free machine $W$ (of some type) that is $c$-incompressible. 
We show that when scaled by the inverse Busy Beaver function, 
both notions of depth are $O(\log |x|)$-close. 
On a prefix-free machine~$W$, both (unscaled) depths are closely related:
Bennett's logical depth of $x$ at significance $c$ is at most $\ld{c+O(1),W}(x)$, 
because any $c$-shortest program $p$ for $x$ is $c+O(1)$-incompressible on $W$. 
On the other hand, by~\cite[Lemma 5.3]{DepthAndReducibility} (attributed to Bennett~\cite{ben88}),
$\ld{c+O(1)}(x)$ is bounded by a computable function of Bennett's logical depth of $x$ with
significance $c$.  Hence, after rescaling with the inverse Busy Beaver function, both notions are
$O(1)$-close. Exchanging prefix-free machine by a plain machine, both depth notions are $O(\log
|x|)$-close; indeed this follows by the same argument as Lemma~\ref{lem:depthIsInvariant} for $W=V$
and replacing $|w_V|$ by $O(\log |x|)$-terms in the proof (since $|K_W(x)-C_U(x)| \le O(\log |x|)$).

\section{Alternative proof of~Theorem~\ref{th:ldepthVsSoph}}
\label{sec:alternativeRelating}

An alternative proof for the second inequality in  Theorem~\ref{th:ldepthVsSoph} is given: 
there exists $e$ such that for all $c$ and $x$ with $|x| \ge e$ we have
\[
    \soph{c + e\log |x|}(x) \le \ld{c}^{bb}(x) + e\log |x|\,.
\]

A {\em prefix stable} machine $V$ is a plain machine 
such that for all strings $p$ and extensions $q$ of $p$: if $p \in \Dom V$ then $q \in \Dom V$ and
$V(p) = V(q)$. For (infinite) sequences $\omega$ let $V(\omega)$ be $V(p)$ if a prefix $p$ of
$\omega$ exists such that $V(p)$ is defined, and undefined otherwise.
For any string or sequence $\omega$, let $0.\omega$ be the real $\sum_i \omega_i 2^{-i}$. 
A prefix stable machine is  {\em left computable}~\cite{EpsteinLevin} 
if for $p$ such that $V(p)$ is defined and for all $q$ such that $0.q \le 0.p$, also $V(q)$ is defined. There are
universal prefix stable machines that are left computable (just rearrange the programs on a
universal machine).
Let $\Omega = \sup\{0.p: V(p) \textnormal{is defined} \}$.

In order to prove the result aforementioned, it is sufficient to show that $\soph{c + 2\log |x|,U}(x) \le \ld{c,W}^{bb}(x) + 2\log |x|$ for large $x$, 
where $W$ is a universal left computable machine. 
Indeed, there exists a universal plain machine $U$ such that 
\[
\ld{c + 2\log |x|,W}^{bb}(x) \le \ld{c,U}^{bb}(x) + O(\log |x|).
\]
(translating plain programs to self-delimiting ones can happen by affecting 
program sizes by at most $O(\log |p|)$ and computation time by a computable function of $|p|$ and
the halting time).

Let $p$ be a program satisfying the conditions in the definition of $\ld{c}^{bb}(x)$.
We show that the initial segment where $p$ and $\Omega$ 
are equal defines a computable function that satisfies the conditions in the definition of sophistication.
More precisely, let $i$ be the length of the common initial segment, then 
$F(d) = V(\Omega_1\dots \Omega_i0d)$ satisfies the conditions.
Note that, $p_1\dots p_i = \Omega_1\dots \Omega_{i-1}0$ and $\Omega_i = 1$ by construction of~$i$.
Thus $F(p_{i+2}\dots p_{|p|}) = x$.
 For any $d$ 
we have $0.\Omega_1\dots \Omega_i0d < \Omega$ and by left computability this is in $\Dom V$, thus $F$ is computable.
It remains to show that $C(F) \le \ld{c,W}(x)+O(\log \ld{c,W}(x))$. 
We show that $C(\Omega_1\dots \Omega_{i-1}) \le \ld{c,W}(x) + O(\log i)$.
In fact, given $i$ and a $t$ that exceeds the computation time of $p$, we can search for the maximal
value $0.w$ for a program $w$ that halts in $t$ computation steps. 
We know that $0.p \le 0.w \le \Omega$, hence we can compute the first $i-1$ bits of $\Omega$ which completes the proof.

\end{document}